\documentclass[journal,draftcls,onecolumn,12pt,twoside]{IEEEtranTCOM}
\normalsize

\ifCLASSINFOpdf
\else
\fi

\usepackage{subfig} 
\usepackage{xcolor}
\usepackage{enumerate}
\usepackage{lineno}
\usepackage{amssymb}
\usepackage{amsmath}
\usepackage{amsthm}
\usepackage{empheq}
\usepackage{caption}
\usepackage{diagbox}
\usepackage{url}
\usepackage{algorithm,algorithmicx}
\usepackage{algcompatible}
\usepackage{lineno} 

\usepackage{graphicx}
\usepackage{caption}
\usepackage{rotating}
\usepackage{bbm}
\usepackage{multirow}
\usepackage{cite}

\newtheorem{theorem}{Theorem}

\newcommand*{\colorboxed}{}
\def\colorboxed#1#{%
  \colorboxedAux{#1}%
}
\newcommand{\comment}[1]{}
\newcommand*{\colorboxedAux}[3]{%
  \begingroup
    \colorlet{cb@saved}{.}%
    \color#1{#2}%
    \boxed{%
      \color{cb@saved}%
      #3%
    }%
  \endgroup
}

\begin{document}
%
\title{On the Fault Tolerant Distributed Data Caching using LDPC Codes in Cellular Networks }

\author{
Elif Haytaoglu, ~Erdi Kaya and ~Suayb S.~Arslan
 \thanks{E.~Haytaoglu is with the Department
of Computer Engineering, Pamukkale University, Denizli,
Turkey e-mail:eacar@pau.edu.tr, Erdi Kaya is with the Department
of Computer Engineering, Ankara University, Ankara, Turkey e-mail:erdik@ankara.edu.tr and S.~S.~Arslan is with Department of Computer Engineering, MEF University, Istanbul, Turkey e-mail:arslans@mef.edu.tr.}
  }
\markboth{}%
{}
%
\maketitle
\begin{abstract}
The base station-mobile device communication traffic has dramatically increased recently due to  mobile data, which in turn heavily overloaded the underlying infrastructure. To decrease Base Station (BS) interaction, intra-cell communication between  local devices, known as Device-to-Device, is utilized for distributed data caching. Nevertheless, due to  the continuous departure of existing nodes and  the arrival of newcomers, the missing cached data may lead to permanent data loss. In this study, we propose and analyze a class of LDPC codes for distributed data caching in cellular networks. Contrary to traditional distributed storage, a novel repair algorithm for LDPC codes is proposed which is designed to exploit the minimal direct BS communication. To assess the versatility of LDPC codes and establish performance comparisons to classic coding techniques, novel theoretical and experimental evaluations are derived. Essentially, the theoretical/numerical results for repair bandwidth cost in presence of  BS are presented in a distributed caching setting. Accordingly, when the gap between the cost of downloading a symbol from  BS  and  from other local network nodes is not dramatically high, we demonstrate that LDPC codes can be considered as a viable fault-tolerance alternative in cellular systems with caching capabilities for both low and high code rates.

\end{abstract}


\vspace{0.5cm}

%
\IEEEpeerreviewmaketitle

\section{Introduction}
In cellular networks, some of the cached files may be stored in a nearby Base Station (BS) as well as in the local memory of end-users. In the latter scheme, Device to Device (D2D) communication can be utilized as an alternative to retrieving data directly from a BS \cite{Asadi} which has been deemed to be essential for 5G and beyond technologies \cite{marsch20185g}. In this communication model, the file is cached in a set of distributed devices and these devices can communicate directly with each other without any BS involvement. This D2D technology not only decreases the communication burden between the BS and distributed nodes but also reduces the access time for retrieving the popular content. Therefore, whenever a mobile device wishes to retrieve data content, it is primarily requested from the nearby devices' caches through D2D communication. In case of data unavailability, nodes request it from the BS at the expense of increased download cost. Thus, if the content is already available on the nearby devices, the content can be accessed without contacting the BS and thereby reducing the traffic and overall access time. The client device may also cache the downloaded data from the BS for further requests that may come from other nearby nodes within the same cell at the expense of using more local storage space.

{The caching strategy in cellular devices is  divided into two major categories based on whether the file is stored in uncoded form \cite{Maddah, shanmugam2013femtocaching} or in coded form \cite{pakkonen2013,Paakkonen15,pedersen,Ji2016}. Regardless of the storage strategy, the data stored in a device can be considered lost whenever that device departs from the cell. One way to prevent the loss of data is the addition of controlled redundancy to the cached content.  This redundancy may simply be achieved for the uncoded data by using traditional replication, i.e., creating multiple identical copies. As for the latter case, a few erasure coding strategies can be used, such as traditional Reed-Solomon codes  (RS) \cite{ReedSolomon}, Locally Repairable Codes (LRC) \cite{locallyrepairablecodes}  and regenerating codes \cite{DimakisGWWR10} which  have already been  used  in the classical distributed storage systems.}

 While the data is stored in the cell, node departures may lead to reduced redundancy and hence the vulnerability for the stored content. To maintain a target level of reliability, data stored on departing nodes need to be regenerated either exactly or functionally within the cell. In the literature, there is a myriad of studies on repairing lost data due to node departures and permanent data loss.  Previous studies generally focus on caching the uncoded data  \cite{Maddah}. However, the most recent studies \cite{pakkonen2013,Paakkonen15,pedersen,Ji2016} mainly analyzed the performance of coded caching systems in which the file is first encoded using a classical erasure code and distributed to the network nodes in the cell. These studies reveal that classic erasure correcting codes (ECC) surpass the other network coding techniques which are specifically designed and optimized for distributed data storage on wired links.


To generate redundancy, generally, a file is first {chopped} into $k$ data symbols which are then encoded into $n$ symbols. When the traditional erasure codes are used, an increase in the values of the parameters $n$ and $k$ leads to higher node repair and decoding complexities which is rarely the case for LDPC codes \cite{Yongmei2015}. 
In addition, these codes can  be a suitable choice for dynamic systems such as cellular networks in which the number of incoming and outgoing nodes is hard to predict. {Hence, LDPC codes become an attractive candidate for distributed caching in cellular networks.} However, to the best of our knowledge, there is no work in the literature on any LDPC code family used  {in the context of} coded caching and in presence of a BS.

{The primary {motivation} in this study is to present that it is possible to reduce repair bandwidth cost using LDPC codes in presence of base station/s. In other words,  utilization of the base station to repair a lost node using LDPC codes  as a coded caching mechanism is one of the noteworthy {contributions} of this study.} {One of the objectives is to improve the data repair functionality of existing LDPC codes such as \cite{Eleftheriou2002} in presence of BS, which was originally proposed for reliable communications. {In particular, LDPC codes are essentially used to provide a reliable communication, with these codes} lost symbols due to communication channel problems are attempted to be recovered. However,  in the coded caching systems, symbol losses may occur due to node departures or node failures. Besides, the existence of a base station changes the nature of the traditional decoding process of LDPC codes. The communication cost of the base station as well as local nodes need to be jointly considered for an optimal system operation. }

To this end, in this paper, we analyzed the node repair cost of a distributed caching system in terms of downloaded symbols from the BS as well as from the neighboring devices in which LDPC codes are used as the fault tolerance mechanism. Furthermore, we presented a novel node repair algorithm tailored to LDPC\footnote{By that rationale, to other sparse graph codes such as fountain codes as well.} codes.  In this study, a specific array LDPC code family given in \cite{Eleftheriou2002}  is used due to its efficient encoding, and its {recovery} performance is compared with that of classic  RS codes \cite{ReedSolomon}, Minimum Bandwidth Regenerating (MBR)  and Minimum Storage Regenerating (MSR) codes \cite{DimakisGWWR10, EMSRHighRate, Rashmi11} using different rates  {which are the well-known and generally accepted codes among their contemporaries.} Our contributions can be summarized as follows:
\begin{itemize}
   
   \item A novel repair algorithm is proposed to reduce the total number of symbols downloaded from the BS using LDPC-based distributed caching.
   \item  The theoretical bounds are derived for the average required number of downloaded symbols from the BS as well as the distributed network nodes in the repair process.
   \item The proposed algorithm and the well known methods in the literature using RS, MBR and MSR codes are numerically compared in a carefully designed simulation environment in terms of the number of downloaded symbols from the BS as well as the helper nodes.
   \item  {We demonstrate the performance gap between the optimum repair strategy achieving the minimum bandwidth cost in LDPC-based distributed caching and our proposed algorithm by exhaustive search.}
\end{itemize}

{Our simulations demonstrate} that the lowest BS communication is achieved by RS and MBR codes with rates $1/2$ and $3/4$, respectively. The lowest D2D communication {occurs when LDPC and MSR codes are employed for the rates  $1/2$ and $3/4$.}  By and large, LDPC codes outperform MSR codes in terms of BS usage both for low and high code rates. { Besides, when  the ratio of {the cost of downloading a symbol from the BS over downloading it from other network nodes} is not very high, LDPC codes can be a very reasonable choice due to their low encoding, decoding (reconstruction) complexities \cite{Eleftheriou2002} and their low storage overhead.}  

The rest of the paper is organized as follows. The related work on coded caching systems employing different erasure codes is summarized in Section II. In Section III, the system model is elaborated. Moreover, the bandwidth costs for regenerating a lost node, which may have multiple symbol failures are derived. In Section IV, an algorithm, named Greepair, for repairing a lost node for generic LDPC code ensambles is elaborated. {In Section V, we present a number of evaluations of the proposed method, which include  notes on comparisons in terms of storage overhead and node repair complexity as well as simulation results. We also provide theoretical support to validate our numerical findings.} Finally,  we conclude our paper in Section VI.

\section{Related Work}
\subsection{Distributed Coded Caching}
The caching methods may differ in the way the data files are stored. In some of these methods, such as found in \cite{Maddah,shanmugam2013femtocaching},  data is stored in uncoded form while in some other forms such as found in {\cite{Paakkonen15,pedersen,pedersen2018optimizing,piemontese2018mds,calis2016maintenance,paakkonen2018coded,li2018double}} data is stored in a coded form using erasure codes in order to secure low storage overhead. However, minimizing the maintenance cost for node departures and the storage cost are amongst the main objectives in all of the past research work.

In cellular networks,  popular media content is repeatedly requested by many users simultaneously, which plays the key reason behind increased data traffic across the network. This issue can be addressed by a distributed storage system based on intelligent caching techniques. To reduce the popular video traffic in wireless networks, the concept of femtocaching was proposed and analyzed ~\cite{shanmugam2013femtocaching, golrezaei2012femtocaching}. In this architecture, small cell access points store the relevant video content in their local caches. The overhead on these points may exceed the data transfer capacity and lead to a bottleneck.  The helpers are equipped with low backhaul capacity and high storage capability to reduce the probability of system bottleneck. The cached data are either kept in uncoded or coded form.

{In  \cite{pakkonen2013,Paakkonen15}, the erasure codes were compared in terms of maintenance costs due to node departure(s) or node failures that  typically occur  in any coded caching system.} Among these studies, particularly in \cite{pakkonen2013},  simple caching is exercised in a setting in which the retrieved file is cached as is as well as in a coded form using two-way replication and regenerating codes. These mechanisms are compared in a simulation setting whereby mobile devices enter in and out of the cell according to the Poisson random process, and the node repairs are performed instantaneously. Moreover, a novel method is proposed to choose either simple caching or {redundant caching} for minimizing energy consumption. {In \cite{Paakkonen15}, the performance of regenerating codes is evaluated in terms of total energy consumption occurred in the node repair process of  the coded caching systems.} Unlike the earlier studies, it is assumed that nodes have unlimited storage capacities. Accordingly, {it is demonstrated that} the coded caching outperforms the uncoded version in terms of energy consumption.  

In \cite{pedersen}, the repair process is initiated repeatedly at periodic intervals. This \textit{lazy} repair process provides a more realistic scenario so that it takes into account live mobile devices as well as the BS within the same cellular network. Accordingly, it was observed that the regenerating codes, which are optimized for low bandwidth usage, did not demonstrate the same performance in such a lazy repair scenario. In fact, MDS codes, which require extensive repair bandwidth, were shown to perform better in terms of bandwidth consumption in a distributed caching setting.

In a more recent study \cite{paakkonen2018coded}, several different caching strategies were investigated; (1) \textit{simple caching} in which only one node in the cell stores a full copy of the file and (2) \textit{caching with redundancy} in which  node(s) requesting the file also cache it using \textit{replication} and \textit{regenerating codes} (MBR and MSR) \cite{DimakisGWWR10}. In that study, the content caching performance of coded and uncoded systems is compared in terms of communication cost. In their system model, it is assumed that the immediate repair is processed upon any node failure in order to keep the number of storage nodes constant. It is found that the multicaching method outperforms the other caching methods only if the file is popular, whereas the replication achieves the best performance when the file is not so popular. Finally, it is concluded that D2D caching is not recommended for the cases where the file is unpopular and storage costs are very high. {We provide Table \ref{tab:1} that systematically classifies aforementioned studies on the basis of coded caching paradigm.}
\begin{table*}[h!]
 \caption{ A classificiation of studies for caching  in mobile networks }
\begin{tabular}{ | p{0.14\textwidth} || p{0.2\textwidth} || p{0.20\textwidth}||p{0.09\textwidth}  ||p{0.20\textwidth} |}
\hline
\textbf{Coding System }& \textbf{Caching Method }& \textbf{Repair/{Decode} Strategy} & \textbf{Stored File} &  \textbf{Aim}\\
\hline
\hline
Maddah et al.  (2013) ~\cite{Maddah} & Uncoded caching & NA & Multiple files & Comparing Global caching vs. local caching \\
\hline
Shanmugam et al. (2013)~\cite{shanmugam2013femtocaching} & Uncoded caching and fountain coded caching & NA &  Multiple files & Minimizing delay through optimum file placement
\\\hline

Paakkonen et al. (2013)~\cite{pakkonen2013} & Simple Caching, {redundant caching} & Instantaneous Repair & Single File  &  Minimizing energy consumption \\ 
\hline
Paakkonen et al. (2015) ~\cite{Paakkonen15} & Simple caching, replication, regenerating codes &  Instantaneous Repair &Single File & Minimizing energy consumption\\\hline
 Pedersen et al. (2016) ~\cite{pedersen}.  & MDS,  replication, regenerating codes, LRC codes & Lazy Repair & Single File  & Communication Cost Analysis  \\ 
\hline

 Pedersen et al. (2018)  {~\cite{pedersen2018optimizing}} & Optimized MDS erasure codes  & Performed if it is necessary when a node request
a file. & Multiple Files & Optimizing caching of content  \\\hline

 Paakkonen et al. (2019)~\cite{paakkonen2018coded} & Simple caching, regenerating codes, Multicaching  & Instantaneous Repair &Multiple Files & Minimizing  communication and storage cost 
 \\
 
 \hline
Li et al. (2018)  ~\cite{li2018double} & Double replication MDS codes & Lazy Repair & Single File & New hybrid code scheme for caching   
 \\
\hline

\end{tabular}

\vspace{-4mm}
\label{tab:1}
\end{table*}

Replication, RS and regenerating codes are applied in many system settings and configurations to secure optimal storage overhead and maintenance costs. To our knowledge, there is no study on the use of sparse graph codes (in particular LDPC codes) for cellular systems, which are mainly constructed for use in the physical layer of noisy communication systems \cite{5GPhysical}. In fact, such codes possess sparse parity check matrices which make them an appealing candidate due to their low decoding (reconstruction) complexity and inherent locality properties. In this study, we observe that these sparse graph codes can be used as an alternative to existing techniques with novel modifications to the traditional decoding/repair processes.




\subsection{ LDPC Codes for Classical Distributed Storage}

 {All the referred erasure codes have been used as a fault tolerance mechanism for distributed storage systems before being applied in the coded caching systems. For instance,  LDPC codes are used for distributed data storage in \cite{wei2014auto}}. In that study, the lost symbols are repaired one by one using different LDPC codes. 
This repair process was compared with replication, RS, and LRC codes \cite{locallyrepairablecodes} in terms of system resource consumption.  According to the presented results, it was observed that the latency due to encoding and decoding processes is far lower than that of RS codes. The trade-off between key criteria such as repair bandwidth, reliability, and storage overhead used in the performance evaluation process of LDPC codes is also partially analyzed in \cite{park2017ldpc}. Unlike RS codes, it is shown that the repair bandwidth of LDPC codes does not change even if the code blocklength increases. According to this study, {right-}regular {(regular check node degree)} LDPC codes are optimal in an average sense and shown to provide lower bandwidth and higher reliability than RS codes.





The study in \cite{wei2014auto}  focuses on repairing lost symbols one at a time (\textit{single symbol repair}) while in a later study \cite{Yongmei2015}, the repair process of  multiple lost symbols all at once named as (\textit{multiple symbol repair}) using the generator matrix's  linearly independent rows  is studied. When there exists an exceeding number of symbol failures, multiple symbol repair {cannot} be practical due to matrix inversion requirements. In the single symbol repair,  \textit{recovery equations} are constructed as follows. A row of the parity check matrix of the LDPC code is selected in which the column corresponding to the lost symbol value is nonzero. Preferably, the row with the minimum row weight is chosen in order to use the lowest bandwidth possible. Then, the symbols corresponding to ones in the selected row but the lost symbol should be collected from the other network nodes. The addition of these symbols over the binary field constitutes the recovery equation of the lost symbol \cite{Yongmei2015}. 
 
  

\section{System Model} \label{sec:3}
Let a cell contain a total of $N$ network nodes at the onset subject to the arrival and the departure of new nodes over time. All nodes in the cell can communicate with each other without BS assistance. For simplicity, let us assume that a single file of size $F$ {symbols} is cached by a total of $m \leq N$ nodes. This file is initially partitioned into $k$ equal size data chunks to be erasure encoded systematically into $n > k$ equal size chunks {when LDPC and RS codes are used}. {When the regenerating codes are used, a file having $F$ symbols are partitioned into pieces having $B$ symbols, and each piece is encoded into $n$ packets such that each of them has $\alpha$ symbols}. Finally, these encoded symbols/packets are distributed over $m$ nodes with $m \leq n$. Each storage node stores exactly the same number of symbols i.e., $n/m$ symbols/{packets} per stripe and $m \mid n$, unless LDPC codes are utilized. In the case of LDPC, nodes  store $\lceil n/m \rceil$ symbols due to the construction constraints as described in \cite{Eleftheriou2002}.  In addition, the raw $F$ {symbols}, as well as the generated redundancy, are stored in the BS, whereas newcomer nodes store no related information regarding the cached file. A summary of the used notation throughout the paper is provided in Table \ref{tab:nomenclature}.
\subsection{Arrival-Departure Model}
Let us define $N_{\mathrm{in}}$ to be current number of nodes in the cell. At the beginning of the operation there exist $N$ nodes, {i.e, $N_{\mathrm{in}}=N$}. However, $N_{\mathrm{in}}$ may change over time due to node arrivals and departures.  It is assumed that these nodes arrive at and depart from the cell according to the Poisson process as assumed in the past literature \cite{pakkonen2013,pedersen}. Our simulation model is constructed according to M/M/$\infty$ queuing model as in \cite{pedersen},{\cite{Pourmandi21}}. {Similar to the model described in \cite{pedersen}},  nodes enter the cell  based on  a Poisson process with  independent and identically distributed (i.i.d.) exponential random inter-arrival times given by the probability density function (pdf) {$f_T(t) =  N\lambda e^{-N\lambda t}$}. Likewise, each node leaves the cell at a rate $\mu$ {(per node)}  and {they are in the cell during } i.i.d. exponential random period given by the pdf $f_T(t) = \mu e^{-\mu t}$ as in \cite{pedersen}.
\begin{figure}[!htb]
    \centering
    \includegraphics[width=0.55\textwidth]{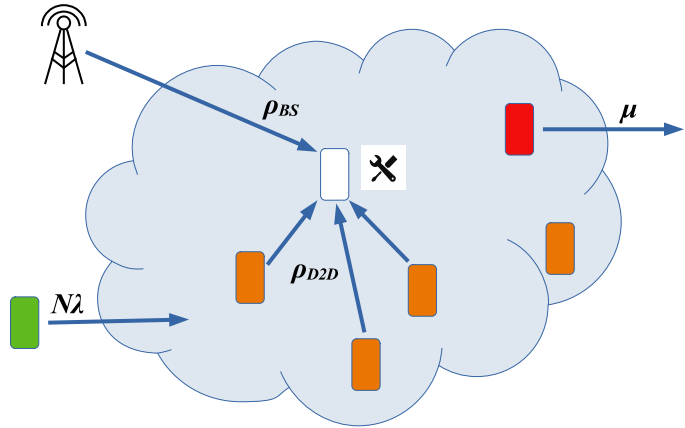}
    \caption{{A wireless network where distributed storage is implemented by mobile devices (nodes).} 
    }
    \label{fig:mesh1}
\end{figure}
 Hence, the probability of the existence of $x$ nodes at a given time can be described by $p({N_{in}=x})=\frac{{({N\lambda}/{\mu})}^{x} e^{-(\frac{N\lambda}{\mu})}}{x!}$ {as stated in \cite{pedersen}}. In order to maintain an average number of $N$ nodes over time,  $\lambda$ is typically selected to be equal to  $\mu$. The typical system model is given in Fig. \ref{fig:mesh1}.

\subsection{Node Repair}
Whenever a storage node leaves the cell, the stored content in that node may be lost unless a precautionary protocol allows pre-replication of the content to another empty node in the cell. To keep the same reliability level of the cached content, the lost data must be repaired. However, in our system model, instead of repairing the lost content immediately, the repair processes are initiated periodically, known as \textit{lazy repair},  at the end of a time interval $\Delta$  \cite{totalrecall}.  Moreover, we assume that there is no node departure or arrival during repair operations. To repair the lost storage content on the departed node, an empty node is randomly selected from the cell and execute a systematic repair process {till it becomes} part of the set of storage nodes. Since the departing node may store more than one symbol due to $ \lceil n/m  \rceil> 1$, repairing the lost content may lead to multiple  symbol recovery operations. Besides, multiple node departures may occur simultaneously within $\Delta$ time. It is assumed that the nodes carrying out the repair operations in the same period do not communicate with each other in order to perform the repair in a totally distributed manner.
\subsection{Analysis of Repair Cost}
In this study, the primary objective is to minimize the cost of node repair bandwidth {according to  the type} and the construction of an error correction code. Thus,  the cost of downloading one symbol from any storage node (represented by $\rho_{\mathrm{D2D}}$) and the cost of downloading one symbol from the BS (represented by $\rho_{\mathrm{BS}}$)  are the main parameters of interest.  According to path loss laws of wireless channels, it is reasonable to assume that the cost of downloading one symbol from a device is lower than that of downloading it from a BS, i.e., $\rho_{\mathrm{D2D}} < \rho_{\mathrm{BS}}$ \cite{Paakkonen15}.  {Moreover, the cost of BS communication can vary dynamically due to a number of reasons such as the frequency of data access  requests, available physical and network conditions, etc. Thus, the ratio $\frac{\rho_{\mathrm{BS}}}{\rho_{\mathrm{D2D}}}$  may not be exactly known prior to the node repair process. {Hence}, it is assumed that the ratio $\frac{\rho_{\mathrm{BS}}}{\rho_{\mathrm{D2D}}}$  is typically higher than the  lower bound at which D2D enabled repair cost is always less costly compared to downloading lost symbols directly from the BS. For instance in a {$(d_v,d_c)$ regular} LDPC codes, this lower bound would be $d_c-1$ due to the inequality $(d_c-1)\rho_{\mathrm{D2D}}<\rho_{\mathrm{BS}}$. Therefore, it is assumed that the node repairs are done with the minimum the BS involvement for all the erasure codes referred below. }
\begin{table*}[!t]
    \centering
    \caption{A summary of symbols/definitions}
    \begin{tabular}{ | p{0.055\textwidth} || p{0.38\textwidth} || p{0.055\textwidth}||p{0.38\textwidth} |}
   \hline
       \textbf{Symbol} & \textbf{Definition} & \textbf{Symbol}& \textbf{Definition} \\    \hline
          \hline
        $m$ & Number of storage nodes within the cell &$N$  & {The initial number of nodes in the cell} \\ \hline  
        $\rho_{\mathrm{D2D}}$& Cost of downloading a symbol from other nodes&$\rho_{\mathrm{BS}}$ & Cost of downloading one symbol from BS \\ \hline $N_{\mathrm{in}}$& Total number of nodes in the cell & $a$&  The  number of  repaired symbols of the lost node\\\hline 
        $F$& Original file size & $l$& Total number of lost symbols in the system\\\hline 
        $\Delta$ & Time interval between two repairs  & $p$ & Probability of a node staying in the cell \\\hline 
        $B_1$ & Number of message symbols to be encoded in one stripe for MBR codes& $B_2$ & Number of message symbols to be encoded in one stripe for MSR codes \\ \hline $t^z$ & Number of  symbols stored in a node using high rate MSR codes & $q(a)$ & Defines the probability of a symbol loss where $a$ symbols have just been repaired in LDPC codes. \\
          \hline
    \end{tabular}
    \label{tab:nomenclature}
\end{table*}
In order to theoretically support our numerical results, to be discussed in Section \ref{sec:sims}, we have presented four cost functions. Namely, we  specified the number of symbols downloaded from the BS  and from the neighboring nodes in the repair process of a lost node for RS \cite{ReedSolomon}, MBR \cite{Rashmi11}, MSR \cite{Rashmi11,EMSRHighRate} and LDPC codes \cite{Eleftheriou2002}.  The outputs of the cost functions define the bandwidth cost for repairing the content of only one single lost node when there may exist one or multiple lost nodes. We set $ n/m$  symbols/packets to be stored in each storage node that takes place in the systems where RS, MSR, and MBR codes are employed whereas $\lceil n/m \rceil$ symbols are stored in a storage node for the systems using LDPC codes. {Let us analyze} the cost functions of each one of these codes {next}. 
\subsubsection{RS codes}
When there are $l$ lost symbols in the system and the number of alive symbols satisfies the condition $n-l \geq k$,  the contents of the lost node can be repaired through decoding by downloading $k$ symbols from the other devices as specified in \cite{Erdi2020}. On the other hand,  if the number of symbols in all alive nodes is not sufficient to complete the repair operation successfully i.e., if the inequalities $n-l < k$ and $k-n+l < n/m$ hold, the repair operation would require downloading symbols not only from these devices but also from the BS. {In this case, the total number of $n-l$ symbols in the alive node is insufficient for any lost symbol repair. Thus, additional download of $k - n+l$ symbols is required from BS to repair $n/m$ lost symbols as in \cite{Erdi2020}. Since the second condition satisfies the inequality $k-n+l < n/m$, in addition to $n-l$ D2D symbols download, $k-n+l$ symbols are downloaded from the BS  to be able to repair $n/m$ symbols. {If $k-n+l \geq n/m$, it is sufficient to download the original lost $n/m$  symbols directly from BS without the need for a D2D communication as seen in \cite{Erdi2020}.}} {In this case, if both BS and local nodes were used, $n-l$ helper symbols from alive nodes and $k-n+l$ helper symbols from the BS would be downloaded, which is always more costly than downloading the lost $n/m$ symbols directly from the BS.} Thus, the cost function for regeneration of a lost node using  RS codes can be given as
\begin{eqnarray}
C_{\mathrm{RS}}(F,m,n,k,l)
=\begin{cases} F\rho_{\mathrm{D2D}}, &\mbox{if } k \leq n-l \\
\frac{F}{k}(k-n+l)\rho_{\mathrm{BS}}+\frac{F}{k}(n-l)\rho_{\mathrm{D2D}}, & \mbox{if } n-l < k < \frac{n}{m} + n - l \\ 
\frac{n}{m}\frac{F\rho_{\mathrm{BS}}}{k}, &\mbox{if } k \geq \frac{n}{m} + n - l \label{eq:crs}\end{cases}
\end{eqnarray}
where $l$ represents the number of lost symbols in the system. Assuming independence, the probability of having {$i$ storage nodes in the cell during $\Delta$ interval can be given by the binomial distribution $b_i(m,p)={m  \choose i} p^i{(1-p)}^{m-i}$ for $0\leq i \leq m$ as assumed in \cite{pedersen}}.  Hence, the total expected cost of repairing lost nodes within $\Delta$ time, i.e., the expected communication cost required by repairing a lost node for the interval $\Delta$ can be simply  expressed as,
\begin{eqnarray}
\mathbb{E}[C_{\mathrm{RS}_{\Delta}}(F,m,n,k)]
= \sum_{i=0}^{m} {m\choose i} p^i (1-p)^{(m-i)} {C_{\mathrm{RS}}} (F,m,n,k,\frac{n}{m}(m-i)) 
\end{eqnarray} 
where the probability of a node remaining in the cell within $\Delta$ is given by $p=e^{-\mu \Delta}$.

\subsubsection{MBR codes}
{A storage node stores $n/m$ packets each consisting of $\alpha$ symbols when regenerating codes are used. Moreover, $\alpha=d$ for the MBR codes as specified in \cite{Rashmi11}.} In the case of MBR codes, since the regenerated {packets (due to the sub-packetization including symbols)} of a lost node can help with the subsequent repair operations of the same lost node, the repair bandwidth differs substantially from that of RS codes. Here, the communication pattern is formed with respect to the relations among $d,n,l$ as in \cite{Erdi2020}. We define $a$ as the number of packets already regenerated during the lost node repair and assume that the inequalities $d>n/m\geq a$ hold. When the number of symbols required to achieve a single node repair, $d$, satisfies the inequality $d \leq n - l$, there would be a sufficient number of symbols in alive devices.
{In this case, to repair the first lost packet, $d$ symbols should be downloaded from other nodes. To repair the second lost packet, the first packet now can be used as a helper  since it is already repaired. Thus, downloading only $d-1$ symbols is sufficient from other nodes to repair the second packet. In such a sequential manner, all $n/m$ packets can be repaired, whose communication cost can be seen in the second case of Eq. (\ref{eq:mbr}).}
On the other hand, when $d>n - l$, the repair operations are carried out by using both BS and other alive nodes. {Here, to repair a lost packet, a total of $n-l$ symbols are downloaded from other nodes, whereas $d-n+l-a$ symbols are downloaded from BS. When $a=d-n+l$, the need for BS communication can be completely eliminated using previously repaired symbols. This can be observed in the first case of  Eq. (\ref{eq:mbr}).}  Briefly, when an $(n,~k,~d)$ MBR code is used in the system, the communication cost of repairing all $n/m$ packet of a lost node,  can be expressed as,
\begin{align} \label{eq:mbr}
 C_{\mathrm{MBR}}(F,m,n,k,d,B_1,l)=  \begin{cases} 
\sum_{a=0}^{d-n+l-1}\Bigl (\frac{F}{B_1}\rho_{\mathrm{D2D}}(n-l)\dots \\  +\frac{F}{B_1}\rho_{\mathrm{BS}}(d-n+l-a) \Bigr)\dots \\  + \sum_{a=d-n+l}^{\frac{n}{m}-1}\frac{F}{B_1}\rho_{\mathrm{D2D}}(d-a), & \mbox{if }  ~d>n-l \\
\sum_{a=0}^{\frac{n}{m}-1}{\frac{F}{B_1}\rho_{\mathrm{D2D}}(d-a)} &\mbox{if }  ~d\leq n-l \} 
\end{cases}
\end{align}
where $B_1= kd-{k \choose 2}$ to achieve the equality $\beta = 1 $ as given in \cite{Rashmi11}. As a result, the expected cost of repairing a lost node during $\Delta$ interval can be given by the equation given below:
\begin{align}
\mathbb{E}[{C_{\mathrm{MBR}_{\Delta}}}(F,m,n,k,d,B_1)] 
=\sum_{i=0}^{m} {m\choose i} p^i (1-p)^{(m-i)} {C_{\mathrm{MBR}}}(F,m,n,k,d,B_1,\frac{n}{m}(m-i))
\end{align}

\subsubsection{MSR codes}  
In the case of MSR codes, the cost function is heavily dependent on the rate and construction of the code.  Initially, a type of low-rate MSR codes is considered based on \cite{Rashmi11}, which are defined for the parameters $[n,k,d\geq  2k-2]$, $\beta = 1$, { $\alpha=d-k+1$}  and $B_2=k(d-k+1)$ as the file size. Assuming $d>n/m\geq a$, {the repair cost function for such low rate MSR codes, namely $C_{\mathrm{MSR}_{\mathrm{LR}}}(F,m,n,k,d,B_2,l)$ is given below. 
\begin{equation} \label{eq:newMSR}
C_{\mathrm{MSR}_{\mathrm{LR}}}(.)=\begin{cases} 
\sum_{a=0}^{\frac{n}{m}-1}{\frac{F}{B_2}\rho_{\mathrm{D2D}}(d-a)}, &\mbox{if }  ~d\leq n-l, \\
\Bigl ( \sum_{a=0}^{d-n+l-1} (\frac{F}{B_2}(\rho_{\mathrm{D2D}}(n-l)+\rho_{\mathrm{BS}}(d-n+l-a)))\mathbbm{1}_A \dots \\ + (\frac{F}{B_2}\rho_{\mathrm{BS}}(d-k+1)){(1-\mathbbm{1}_A)}\Bigr ) \dots\\
+ \Bigl ( \sum_{a=d-n+l}^{\frac{n}{m}-1}\frac{F}{B_2}\rho_{\mathrm{D2D}}(d-a)\Bigr) & \mbox{otherwise }  
\end{cases}
\end{equation} where $A$ is the condition of $d-n+l-a<d-k+1$ and $\mathbbm{1}_A$ is an indicator function which is 1 if the condition $A$ is true, and otherwise it is 0. The first case defines the cost of repair when only local nodes are used, and the  second case defines the cost of using either both BS and local nodes or BS only.  Notice when $d-n+l-a\geq d-k+1$ and $d>n-l$, downloading the original lost symbols from BS is always less costly than using a hybrid communication model, (i.e., involving both BS and the local nodes at the same time). Thus, the second line of the second case defines the cost of downloading the original lost symbols (in a packet) directly from the BS, which is similar to the third case of $C_{\mathrm{RS}}(.)$.} The  expected cost {expression} of low rate MSR codes is omitted due to its similarity with $\mathbb{E}[ {C_{\mathrm{MBR}_{\Delta}}}(F,m,n,k,d,B_2)]$.

 The cost function of high rate MSR codes is {based on the construction given in \cite{EMSRHighRate} where the codes are not only high rate but also long-length and systematic.} In this case, the costs are derived according to whether the lost nodes are systematic or non--systematic {due to the constructions specified in \cite{EMSRHighRate}}. Besides, $n/m=1$ is assumed.  The  repair  cost function of systematic nodes   for  the  high  rate  MSR  codes is given by
{
\begin{align}\label{eq:MSRHR} C_{\mathrm{MSR}_{{\mathrm{HR}}_{\mathrm{sys}}}}(F,z,n,d,t,B_2,l)= \begin{cases} 
 t^{z-1}d\frac{F}{B_2}\rho_{\mathrm{D2D}}, & \mbox{if }  ~d \leq n-l , \\
  t^{z-1}(n-l)\frac{F}{B_2}\rho_{\mathrm{D2D}}+ (d-n+l)t^{z-1}\frac{F}{B_2}\rho_{\mathrm{BS}}, &\mbox{if }  ~d-n+l < t, \\
 t^{z}\frac{F}{B_2}\rho_{\mathrm{BS}}, &\mbox{otherwise}
\end{cases}
\end{align}
{As in RS codes, the cases represent node repairs  using only local nodes, using both nodes and the BS and using the BS only, respectively.}} {{The repair operation of a non--systematic node is carried out entirely with  the BS communication. Thus, $C_{\mathrm{MSR_{{HR}_{non-sys}}}}(F,z,n,d,t,B_2,l)$  is equal to the third case of Eq. (\ref{eq:MSRHR}) where $n = (t+1)z+t,d = n-1,k = (t+1)z, \alpha=t^z$, $B_2=kt^z$ for $t, ~z\in \mathcal{Z}^+$ as given in \cite{EMSRHighRate}}.}
Thus, the expected cost per lost node during $\Delta$ for the systems using high rate MSR codes is given by
\begin{align}
\mathbb{E}[ {C_{\mathrm{{MSR}_{HR}}_{\Delta}}}(F,z,n,d,t,B_2))] &= \sum_{i=0}^{m} {m\choose i} p^i (1-p)^{(m-i)} \Bigl( \frac{k}{n}C_{\mathrm{MSR_{{HR}_{sys}}}}(F,z,n,d,t,B_2, (m-i)) \nonumber \\  & \ \ \ \ \ \ \ +\frac{t}{n}C_{\mathrm{MSR_{{HR}_{nonsys}}}}(F,z,n,d,t,B_2, (m-i)) \Bigr)
\end{align}

\subsubsection{LDPC codes}
For LDPC codes, the expected cost function for {repairing} a lost node using a $(d_v,d_c)$ regular LDPC code can be upper bounded by,
\begin{align}
\mathbb{E}[{C_{\mathrm{LDPC}}}(F,m,k,n,l)] &\leq  
       \sum_{a=0}^{{\lceil \frac{n}{m} \rceil }-1 }  \left ( (d_c-1)d_v{(1-q(a))}^{d_c-1}\frac{F}{k}\rho_{\mathrm{D2D}}+ {(1-(1-q(a))^{d_c-1}})\frac{F}{k}\rho_{\mathrm{BS}}\right ) \\
       & = \sum_{a=0}^{{\lceil \frac{n}{m} \rceil -1} } \frac{F}{k}  \left (  \rho_{\mathrm{BS}} + (1-q(a))^{d_c-1} \Big((d_c-1)d_v\rho_{\mathrm{\mathrm{D2D}}} - \rho_{\mathrm{BS}}\Big) \right)
\label{eq:expectedLDPC1}
\end{align}
where $q(a)=\frac{l-a}{n-a}$ is the loss probability of a symbol with $a \leq l$ and $l$ can be derived using $b_i(m,p)$. For the regular LDPC codes, any recovery equation has at most $d_c-1$ terms and there exist $d_v$ alternatives for one lost symbol. Thus, if all of  $d_c-1$ symbols can be retrieved from the alive node(s), the lost symbol can be repaired with the cost $(d_c-1)\rho_{\mathrm{D2D}}$, {which can occur with probability  ${(1-q(a))}^{d_c-1}$. Since we have $d_v$ recovery equation alternatives, this probability can be loosely upper bounded by $d_v{(1-q(a))}^{d_c-1}$.}


Nevertheless, when there are no recovery equations whose remaining $d_c-1$ elements are all alive for a lost symbol, BS remains the only option for the repair. Notice that the worst case assumed in the derivation of the upper bound happens when none of $d_v$ equation {alternatives} can not repair the lost symbol.  Here, the probability of BS involvement can be upper bounded by ${(1-(1-q(a))^{d_c-1}})$. 
As the required symbols of recovery equations of different symbols may intersect, the cost function described in Eq. (\ref{eq:expectedLDPC1}) above gives an upper bound. Since the intersections between the recovery equations of different lost nodes rarely happen for long blocklengths, this bound becomes tighter as the blocklength increases.   Thus, the expected cost of repair for LDPC codes during the time frame $\Delta$ will be given by
\begin{align}\label{eq:LDPC}
  \mathbb{E}[{C_{\mathrm{LDPC}_{\Delta}}}(F,m,k,n)] &= \sum_{i=0}^{m}{m\choose i} p^i (1-p)^{(m-i)}\mathbb{E}[{C_{\mathrm{LDPC}}}(F,m,k,n,\frac{n}{m}(m-i))]. 
\end{align}

 \section{Problem Definition}
A single symbol repair operation of  LDPC codes substantially changes when nodes store multiple symbols as well as BS can be utilized by downloading some of the lost symbols on the downlink. For instance, let us consider a $(8,4)$ LDPC code (parity check matrix of which is given below) and the codeword symbols are distributed to four nodes such that $i$-th node caches the symbols $\{s_{(2(i-1)+1)}, s_{2(i-1)+2) }\}$. In addition, assume that there exists a BS that stores all the codeword symbols. Suppose now the first and the second nodes are lost. In an independent repair scenario, each node downloads helper symbols independent of each other.

\begin{equation}\label{eq:parityCheck2}
\textbf{H} = \begin{bmatrix}
       0 &0 &1& 0& 1& 1&0&0 \\
       1 &0 &1& 0& 0& 0&1&0\\
       1 &1 &0& 1& 0& 0&0&0\\
       0 &1 &0& 1& 0& 0&1&1
\end{bmatrix}
\vspace{0.4cm}
\end{equation}

To repair the lost data in the first node,  $\{s_1,s_2\}$ should be repaired.  { {Based on \textbf{H}}, There are two  recovery equation alternatives for repairing $s_1$ which are $s_1=s_3+s_7$ and $s_1=s_2+s_4$.  These equations can be represented by the set $r_{i,j}$ where $r_{i,j}$ represents the $j$-th recovery equation option for repairing symbol $s_i$. Here, $r_{1,1}=\{s_3,s_7\}$ and $r_{1,2}=\{s_2,s_4\}$}. The number of symbols downloaded from the BS would be $1$ if $r_{1,1}$ is used  and this number would be $2$ if  $r_{1,2}$ is used.
{Also, there are  two possible recovery equations to repair $s_2$  which can be identified by $r_{2,1}=\{s_1,s_4\}$ and $r_{2,2}=\{s_4,s_7,s_8\}$ whose BS communication would require retrieving $2$ and $1$ symbols, respectively}.  Let us choose $r_{1,1}$ and $r_{2,2}$ to repair $s_1$  and $s_2$,  respectively. Repairing $s_1$ and $s_2$ using these equations entails the download of symbols $s_3$ and $s_4$  from the BS. However, if $r_{2,2}$ and $r_{1,2}$ were used instead, downloading only $s_4$ from the BS would have been sufficient. Since $s_2$ is repaired first, it can be used to solve $s_1$ using $r_{1,2}$. As can be seen, the number of symbols downloaded from the BS is impacted by the order of recovery equations used. In fact, the problem of repairing lost nodes with the minimum BS interaction can be shown to be  a special case of an NP-Complete problem \cite{HaytaogluArxiv}. 

\subsection{ The Dynamic Minimum Weighted Set Cover Problem and LDPC Repair}

In this subsection, we show how to transform LDPC node repair process as a dynamic minimum weighted subset cover problem. Before the main discussion, four different problems, namely,  the minimum  subset cover problem, the minimum weighted set cover problem,  the minimum weighted subset cover and finally  the minimum dynamic weighted subset are described.

Let us commence with the variation of the minimum set cover problem called the minimum subset cover problem in which there exists $n$ different sets, $S=\{\xi_i \mid i \in Z^+ \ s.t. \ 1 \leq i \leq n \}$ where $\bigcup_i \xi_i ~(\forall ~\xi_i \in S)$ constitutes the universal set $U$.  The  minimum subset cover problem is about finding the minimum number of sets that cover the set $E$ where  $E \subseteq U$ \cite{Kleinberg}. 

Now let us consider a variation of the minimum set cover problem in which each of  $\xi_i$s is associated with a weight $w_i \in  \mathbb R^{+} \cup \{0\}  $. In this problem, it is required to find  a set of $\xi_i$s, namely $S'$ such that  $\bigcup_i \xi_i ~(\forall ~\xi_i \in S')$ constitutes the universal set $U$ as well as the total weight is minimum.  This problem is named as  the minimum weighted set  cover problem in the literature and it is shown to be NP-Complete \cite{Chvatal1979}.

The combination of the problems above leads to yet another problem which is named as \textit{the minimum weighted subset cover problem} where the selected sets should now cover $E \subseteq U$ instead of $U$ itself. Aparently, It is not hard to show that this problem  is also NP-Complete using similar arguments in \cite{Chvatal1979}.

A novel problem appears when the costs of the sets can be changed upon selecting any other set in $S'$. In this  problem, the weights of the sets are changed dynamically and contingent upon previous selections of $\xi_i$s i.e., the selection order of the sets affects the  costs of the unselected sets. More specifically, the initial weight for set $\xi_i$ is calculated as
$w_i=|\xi_i|+c_i,  c_i \in \mathbb{R}^{+} \cup \{0\}$, however, after any set $\xi_y$ is selected, the costs are updated as below.
\begin{eqnarray}
w_i=|\xi_i  \setminus  ( \xi_y \cap E^{\prime})|+c_i, \ \ \ c_i \in \mathbb{R}^{+} \cup \{0\},~E^{\prime}\subseteq U,~E \subseteq U  \label{eqnupdate}  
\end{eqnarray}

We name this particular variation of the \textit{minimum weighted subset cover problem} as \textit{the minimum dynamic weighted subset cover problem}. In essence, the main objective of \textit{minimum dynamic weighted subset cover problem} is to find an ordered set of  $\xi_i$s which would  cover the subset $E$ and whose total weight achieves the global minimum.

Before discussing the  \textit{minimum weighted subset cover problem}, let us recall the requirements of a problem to be NP-complete. 
\begin{theorem}
To prove that a problem $X$ is a member of NP-Complete problem set, three properties given below must hold at the same time. \\
 \begin{minipage}[t]{\linewidth}
\begin{enumerate}[I.]
    \item Prove $ X$ is in NP. 
    \item Choose a well known NP-complete problem Y. 
    \item Prove that Y $\leq_P X$   which means X is pollynomially reducible to  problem Y. 
\end{enumerate}
\end{minipage}
\end{theorem}
\begin{proof}
The proof and details can be found in \cite{Kleinberg}.
\end{proof}
\begin{theorem}\label{teo:3}\textit{The minimum dynamic weighted subset cover problem} is an NP-Complete problem.
\end{theorem} 
\begin{proof}
Based on the properties given in Theorem 1, we examine them one by one for the  minimum dynamic weighted subset cover problem.

\begin{minipage}[t]{\linewidth}

\begin{enumerate}[I.]

\item  When we are given a set of sets, it is known to take polynomial time in terms of $O(|U|)$ to be able to check whether the union of these sets covers the set $E$.

\item  Let us choose the minimum weighted subset cover problem which is known to be NP-complete.

\item   If there is an algorithm for the minimum dynamic weighted subset cover problem, it can also be used to solve the minimum weighted subset cover problem for any instance when $E^{\prime}=\emptyset$ which ensures that costs remain constant.


\end{enumerate}

\end{minipage}
\end{proof}



\begin{theorem}
An instance of a possible solution in LDPC repair problem can be applied in polynomial time.
\end{theorem}
\begin{proof}
When there are $l$ different lost symbols, an instance solution is comprised of a permutation of at most $l$ different recovery equations.   Regardless of the order of the permutation, any lost symbol requires at most $(h-1)$ binary additions where $h$ is the maximum row degree  in the parity check matrix. The number of additions (XOR operations) will be further decreased in case some of the recovery equations contain symbols that have been already repaired. Since, generally, $h\ll n$, the solution can be applied in $O(lh) \in O(n)$ time, but at most in $O(n
^2)$ .
\end{proof}

\begin{figure}[!t] 
	\centering
\includegraphics[width=0.65\textwidth]{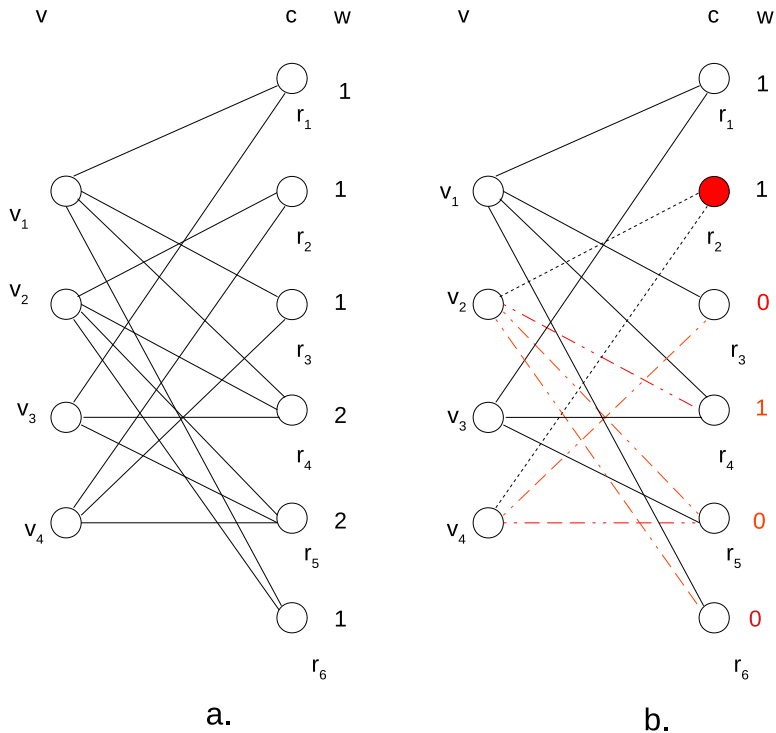} 
\caption{The tanner graph representing dynamic subset cover problem.}\label{tan}
\end{figure}

\begin{theorem}
  LDPC repair problem with the minimum BS intervention is a special class of  the dynamic weighted subset cover problem with $E=E^{\prime}$.
\end{theorem}

\begin{proof}
Let us define $\mathcal{L}$ as the set that consists of  all the lost symbols whereas  $\mathcal{L}_1$ consists of the lost symbols having at least one recovery equation which does not require any base station communication. In Figure \ref{tan}, left hand side nodes (variable nodes) represent the lost symbols in $\mathcal{L}\setminus \mathcal{L}_1$ while the right hand side nodes  (check nodes) represent the recovery equations, namely $r_i$ in the given Tanner graph. Each $r_i$  can be considered as a set of variable nodes induced by the corresponding edges in the Tanner graph. For instance, $r_2$ having degree $d_{r_2}=2$ covers $v_2$ and $v_4$. Let us define $R_2$ to be $\bigcup r_i$ for all $r_i$ that would be a recovery equation for a lost symbol in $\mathcal{L}\setminus \mathcal{L}_1$. without loss of generality, each equation $r_i$ has a weight equivalent/proportional to the value of $d_{r_i}-1$. However, after $r_2$ is selected first the cost as much as $d_{r_2}-1$ is paid, then the edges between other check nodes and the variable nodes connected to $r_2$ are removed from the graph  (see Figure \ref{tan}-b).  Based on this description,  we finally set the inputs as  $E=\mathcal{L}\setminus \mathcal{L}_1$,  $E'={\mathcal{L}\setminus \mathcal{L}_1}$, $U=\mathcal{L}$ and $\bigcup \xi_i= \bigcup r_i$ then it is not hard to see that the optimum  LDPC repair problem turns into a dynamic weigthed set cover problem with $E=E^{\prime}$.
\end{proof}
 
  \section{An Algorithm for Efficient Regeneration of Multiple Symbols}
{Hence, it is crucial to design an efficient algorithm for this repair problem. To this end, a novel greedy algorithm named  \textit{Greepair} is proposed, in which the lost symbols are repaired with minimal BS assistance.}  More specifically,  Greepair algorithm is a two-phase iterative algorithm as detailed in Algorithm \ref{alg1} and Algorithm \ref{alg2} where  $\mathcal L$ contains all symbol indices stored in the lost node and {$\mathcal G$ contains  all lost symbol indices in the storage system, i.e., it also includes the indices of the symbols stored in the other lost nodes}. The algorithm iteratively repairs lost symbols and then removes  their indices from $\mathcal L$  until there remain no symbol indices. In the first phase of the algorithm, the repair operations  that do not require any BS communication are performed. The indices of the lost symbols that can be repaired in this way are stored in a set $\mathcal L_1$. In the second phase, symbols whose repairs involve BS interaction are repaired. 

\textit{{1) The First Phase of the Algorithm:}}
{This phase is a variant of the peeling decoding \cite{peelingLuby} process of LDPC codes  which is originally proposed for reliable communications over erasure channels. In this phase, potentially a maximum number of symbols are attempted to be repaired without BS communication. At the same time,  these local repair operations are desired to be performed using minimal helper symbol downloads from other nodes.}

The algorithm commences with defining recovery equations that correspond to the lost symbol indices whose repairs do not necessitate any BS communication using the parity check matrix of the code. {In other words, in the first phase, recovery equations having minimum cardinality for symbols in $\mathcal L_1$ are identified. These recovery equations are stored as an array, namely $\mathcal{R}_1$, in which each recovery equation  contains helper symbol indices for the corresponding lost symbol in $\mathcal L_1$}. {Since $\mathcal{R}_1$ will be updated after each symbol repair, to keep the original forms of the recovery equations, before starting the symbol repair operations, $\mathcal{R}_1$ is copied to $\mathcal{R}_1''$.} {Afterwards, equations with the minimum cardinality, i.e. requiring minimum helper symbols, are iteratively selected from $\mathcal{R}_1$ for repairing each lost symbol.} {In the meantime, another equation set called $\mathcal{R}_2$ is constructed which will be used in the next phase of the algorithm.} {The set $\mathcal{R}_2$ holds  all recovery equations for lost symbols whose repair processes requires BS communication.} Likewise, each equation in ${R}_2$ consists of the indices of the symbols that take place in recovery equations for the lost symbols in $\mathcal L \setminus \mathcal L_1$. Similar to $\mathcal{R}_1''$;  $\mathcal{R}_2''$ initializes with the copy of $\mathcal{R}_2$. 

\begin{algorithm}[t!]
\caption{Greepair Algorithm  The First Phase} \label{alg1}
{\textbf{input} $\mathcal G$, $\mathcal L$, $\mathcal L_1$ }~$\triangleright$ {$\mathcal G$ is the set of the lost symbols in other lost nodes } \vspace{-4mm} \\
{\textbf{output} $\mathcal{R}_2, \mathcal{R}_2''$}  \vspace{-4mm}\\
{\textbf{function} \textbf{FirstPhase}($\mathcal L$,~$\mathcal L_1$)} \vspace{-4mm}
\begin{algorithmic}[1]
\STATE {construct ~$\mathcal{R}_1 $ and $\mathcal{R}_2 $} \vspace{-4mm}
\STATE {sort ~$\mathcal{R}_1 $ and $\mathcal{R}_2 $ according to recovery equation set sizes in ascending order } \vspace{-4mm}
\STATE {$\mathcal{R}_1'' \gets \mathcal{R}_1$,~$\mathcal{R}_2''\gets \mathcal{R}_2$} \vspace{-4mm}
\WHILE {$ \mathcal{R}_1\neq \emptyset \wedge \mathcal{L}_1 \neq \emptyset $ } \vspace{-4mm}
\STATE {$R \gets  \mathcal{R}_1''[0]$ } \vspace{-4mm}
\STATE{repair symbol $\mathcal F(R)$ using $R$} \vspace{-4mm}
\STATE {$\mathcal L \gets \mathcal L \setminus \mathcal F(R)$ } \vspace{-4mm}
\STATE {$\mathcal L_1 \gets \mathcal L_1 
\setminus \mathcal F(R)$} \vspace{-4mm}

\STATE {$\mathcal{R}_1  \gets \mathcal{R}_1 \setminus r , \forall   ~r \in \mathcal{R}_1 \wedge \mathcal F(R) = \mathcal{F}(r) $}\vspace{-4mm}
\STATE {$\mathcal{R}_1''  \gets \mathcal{R}_1'' \setminus r , \forall   ~r \in \mathcal{R}_1'' \wedge \mathcal F(R) = \mathcal{F}(r) $} \vspace{-4mm}
\STATE {$r \gets r \setminus  \{\mathcal  F(R) \cup R\}$ $,\forall ~r \in \mathcal{R}_1   $} \vspace{-4mm}
\STATE {$r \gets r \setminus \{ \mathcal F(R) \cup R \}, \forall ~r \in \mathcal{R}_2$} \vspace{-4mm}
\WHILE{$\exists ~ r \mid  r \cap \{  \{ \mathcal  L \setminus \mathcal L_1 \} \cup \mathcal G \}=\emptyset \wedge r \in \mathcal{R}_2$}\vspace{-4mm}
\STATE {$idx \gets i \mid \mathcal{R}_2[i] = r$} \vspace{-4mm}
\STATE {$\mathcal{R}_2 \gets \mathcal{R}_2 \setminus r$  } \vspace{-4mm}
\STATE {$\mathcal{R}_1'' \gets  \mathcal{R}_1'' \cup \mathcal{R}_2''[idx]$ }\vspace{-4mm}
\STATE {$\mathcal{R}_2'' \gets \mathcal{R}_2'' \setminus \mathcal{R}_2''[idx]$  }\vspace{-4mm}
\STATE {$\mathcal{R}_1 \gets  \mathcal{R}_1 \cup r$,$~\mathcal{L}_1 \gets  \mathcal{L}_1 \cup \mathcal F (r)$ } \vspace{-4mm}
\STATE {$\mathcal{R}_2  \gets \mathcal{R}_2 \setminus r_i , \forall   ~r_i \in \mathcal{R}_2 \wedge \mathcal F(r_i) = \mathcal{F}(r) $}\vspace{-4mm}
\STATE  {$\mathcal{R}_2''  \gets \mathcal{R}_2'' \setminus r_i , \forall   ~r_i \in \mathcal{R}_2'' \wedge \mathcal F(r_i) = \mathcal{F}(r) $} \vspace{-4mm}
\ENDWHILE \vspace{-4mm}
\STATE  {$\mathcal{R}_1'' \gets \mathcal{R}_1'' \setminus  R$ } \vspace{-4mm}
\STATE  {$\mathcal{R}_1 \gets \mathcal{R}_1 \setminus   \mathcal{R}_1[0]$ } \vspace{-4mm}
\STATE {Update orders of $\mathcal{R}_1$ and $\mathcal{R}_2$  with respect to the cardinalities of recovery equations} \vspace{-4mm}
\STATE {Reorder $\mathcal{R}_1''$ and $\mathcal{R}_2''$ according to $\mathcal{R}_1$ and $\mathcal{R}_2$, respectively.} \vspace{-4mm}
\ENDWHILE \vspace{-4mm}
\STATE { \textbf{SecondPhase}($\mathcal G$,~$\mathcal L$,~$\mathcal{R}_2 $,~$\mathcal{R}_2''$)} 
\end{algorithmic}
\end{algorithm}
\begin{algorithm}
\caption{Greepair Algorithm The Second Phase}  \label{alg2}
{\textbf{input} $\mathcal G$, $\mathcal L$,~ $\mathcal{R}_2 $,~$\mathcal{R}_2''$}  \vspace{-4mm}\\
{\textbf{output} $BSScore$}  \vspace{-4mm}\\
 { \textbf{function} \textbf{SecondPhase}($\mathcal L$, $\mathcal{R}_2 $,$\mathcal{R}_2''$)} \vspace{-4mm}
\begin{algorithmic}[1]
\WHILE{$\mathcal{L} \neq \emptyset$} \vspace{-4mm}
\STATE {$R \gets \mathcal{R}_2''[0]$} \vspace{-4mm}
\STATE {$\mathcal{R}^{'} \gets \mathcal{R}_2[0]$} \vspace{-4mm}
\IF {$|\mathcal{R}^{'}  \cap \{\mathcal L\cup \mathcal G  \}|\neq 0$} \vspace{-4mm}
\STATE {download $F(\mathcal{R}^{'})$ from BS } \vspace{-4mm}
\STATE $BSScore=BSScore+ 1$\vspace{-4mm}
\ELSE \vspace{-4mm}
\STATE{repair the lost symbol using $F(\mathcal{R}^{'})$ 
without BS} \vspace{-4mm}
\STATE  {$r  \gets r \setminus \mathcal{R}^{'} , \forall ~r \in \mathcal{R}_2$}  \vspace{-4mm}
\ENDIF \vspace{-4mm}
\STATE $\mathcal L \gets \{\mathcal L \setminus \mathcal{F(R)} \} $\vspace{-4mm}
\STATE {$r \gets r \setminus \left \{ \mathcal{F(R)} \right \} \forall ~r \in \mathcal{R}_2 $ } \vspace{-4mm}
\STATE {$\mathcal{R}_2  \gets \mathcal{R}_2 \setminus r , \forall ~r \in \mathcal{R}_2 \wedge \mathcal{F(R)} = \mathcal{F}(r)  $} \vspace{-4mm}
\STATE {$\mathcal{R}_2''  \gets \mathcal{R}_2'' \setminus r , \forall ~r \in \mathcal{R}_2'' \wedge \mathcal{F(R)} = \mathcal{F}(r)  $}\vspace{-4mm} 
\STATE {Update the order of $\mathcal{R}_2$   with respect to their BS needs and correlation.} \vspace{-4mm}
\STATE {reorder $\mathcal{R}_2''$ according to  $\mathcal{R}_2$} \vspace{-4mm}
\ENDWHILE 
\end{algorithmic}
\end{algorithm}
Next, {the equations} in sets  $\mathcal{R}_1$ and $\mathcal{R}_2$ are sorted in ascending order based on their cardinalities. {Both of the two phases  use a function  $\mathcal{F}(X)$  whose domain is the set of recovery equations, and the codomain is the set of the lost symbols. This function is defined for mapping a recovery equation to the lost symbol that it  repairs.} 
 Whenever a lost symbol $s_i$ is successfully repaired, it is removed from the sets $\mathcal L$ and ~$\mathcal L_1$. Since there is no need for any equation $r$ such that $\mathcal{F}(r)=s_i$, $r$ is removed from $ \mathcal{R}_1$ as well as from $ \mathcal{R}_1''$ {(see line 9 and line 10 in Algorithm \ref{alg1})}.  To identify the next {repair}, the helper symbols as well as the previously repaired symbols are removed from the {equations} in  $\mathcal{R}_1$ and $\mathcal{R}_2$  {(see line 11 and line 12 in Algorithm \ref{alg1})}. Afterwards,  necessary updates are conducted based on whether the repaired symbol eliminates the need of contacting the BS for any equation in ${R}_2$. If there exists such an equation, {$ r_x = \mathcal{R}_2[i], i \in \{0,1, \hdots |\mathcal{R}_2|-1\}$} where {{$\mathcal{R}_2[i]$ is the $({i+1})$-th equation of $\mathcal{R}_2$}},  it is removed from $\mathcal{R}_2$ (the corresponding equation is removed from $\mathcal{R}_2''$) and this equation and  $\mathcal F(r_x)$ are added to  $\mathcal{R}_1$ (and the corresponding equation is added to $\mathcal{R}_1''$) and  $\mathcal L_1$, respectively.  {All equations} $r_i$ satisfying $\mathcal F(r_x)=\mathcal F(r_i)$ are also removed from $\mathcal{R}_2$ and $\mathcal{R}_2''$. These operations are repeated greedily until no equation is left in $\mathcal{R}_1$.
 
\textit{{2) The Second Phase of the Algorithm:}}
After the first phase, the remaining lost symbols will be repaired using the equations in $\mathcal{R}_2$, which themselves may have lost symbols to be repaired with the BS assistance.  As long as there are lost symbols in $\mathcal L$, the next reduced equation involving the fewest BS communication rounds from the set $\mathcal{R}_2$ is selected. {When there are multiple equations ensuring the minimum use of BS, the equation of a lost symbol $\mathcal F(r)$  with the highest $|\{\bigcup j | ~r_j \cap \mathcal F(r) \neq \emptyset, \forall r_j \in \mathcal{R}_2 \}| $ is selected.
To put it another way, the equation of a lost symbol which is used most frequently in other recovery equations in $\mathcal{R}_2$ is selected.} If the selected equation requires any BS contact, the lost symbol is downloaded from the BS directly. Next, necessary updates are performed on the sets $\mathcal{R}_2$ and $\mathcal{R}_2''$ while preserving the cardinality of sets in ascending order based on the required  BS usage.
Note that repaired symbols can be used as helper symbols in the next phase.  The fact that no repairs can be carried out without any BS interaction does not necessarily mean that no repair is possible without any BS interaction in the subsequent iterations of the algorithm. In fact, the next repair operation may be performed using the previously repaired (potentially with the help of BS) symbol(s) as a helper symbol. 

\section{{Numerical} Evaluations}
\subsection{Notes on Comparisons}
{We shall investigate the same-rate codes, i.e., the same $\frac{k}{n}$, for RS, MBR, and LDPC codes as performed in \cite{park2017ldpc}. Due to the graph construction of LDPC codes, the parameter $n$, the blocklength for LDPC codes are selected to be greater than that of RS, MBR, and MSR codes.  Moreover, in our simulations, the communication cost is calculated as follows: $\gamma_{c_{\Delta}}=\frac{\rho_{\mathrm{D2D}}\tau_{\Delta}+\rho_{\mathrm{BS}}\phi_{\Delta}}{F_\mathrm{{RS_{node}}}} = \frac{C_{\mathrm{code}}(.)}{F_\mathrm{{RS_{node}}}}$ where $\tau_{\Delta}$ and $\phi_{\Delta}$ represent the number of symbols downloaded from other nodes and  BS  to repair a lost node for a given time interval  $\Delta$, respectively, {and $C_{\mathrm{code}}(.)$ is the cost function of repairing a lost node for the  erasure code used, e.g., if RS is used, it is $C_{\mathrm{RS}}(.)$ or if MBR is used, it becomes $C_{\mathrm{MBR}}(.)$, etc.} Actually, the communication cost results for repairing a single lost node is obtained first and then, for the sake of simplicity, all the cost results are normalized by the size of a RS node which is $F_\mathrm{{RS_{node}}}=Fn_{\mathrm{RS}}/(k_{\mathrm{RS}}m_{\mathrm{RS}})$. 
Notice that, due to construction constraints for high rate MSR codes, we are bound to select similar code rates as that of the other codes. }
\subsubsection{Storage overhead}
{In RS and LDPC codes, to store a file of size $F$ {symbols}  using $R=k/n$ code, a total of $F/R$ {symbols}   should be stored in the whole system. Furthermore, a storage node should store $\frac{F}{Rm}$  {symbols}  where $m$ is the number of storage nodes.}

{To store a file of size $F$ {symbols}  using an $(n,k,d)$  MSR code, a total of $\frac{Fn(d-k+1)}{k(d-k+1)}$ {symbols}  should be stored in the system and a storage node would store $\frac{Fn(d-k+1)}{mk(d-k+1)}$ {symbols}  and the number of symbols downloaded from a helper node in the repair process is $\beta=1$. Assume that we use the same 
$R$ and $m$ with that of RS and LDPC codes and to be able to have the same storage overhead, the following equation should satisfy}
$\frac{Fn(d-k+1)}{mk(d-k+1)}=\frac{Fn}{km}$
which can hold for any $d$ and $k$.
Hence, the storage overhead of MSR codes is equal to that of RS and LDPC codes when these codes have the same rates and the same $m$. For MBR codes, in order to generate the same amount of overhead (with respect to RS codes and LDPC codes) we would need to choose $(n^\prime,k^\prime,d^\prime)$ such that 
$\frac{2Fd^\prime n^\prime}{mk^\prime(2d^\prime-k^\prime+1)} = \frac{Fn}{mk}$.
{Since both codes use the same $m$, the same level of reliability implies that both have to have the same ratio $k/n = k^\prime/n^\prime$, then we will immediately have $2d^\prime = 2d^\prime - k^\prime+1$. However, this would imply $k^\prime=1$ which is not practical for any circumstances.} Consequently, we need to give away from the reliability if we would like to ensure the same storage overhead. Thus, for a given storage overhead, it is impossible to get the same reliability. For the given reliability, it is impossible to get the same redundancy (storage overhead). {In our further comparisons, the storage overhead of MBR codes will be naturally larger than that of RS, LDPC and MSR codes for a given level of data reliability.}

\subsubsection{Node repair complexity}
{{Here we provide a brief summary of node repair complexities for all the codes compared. We realize that since the repair process is closely related to the code construction method at hand, different repair complexities can be obtained using different constructions in the literature.}
In {$(d_v,d_c)$ regular} LDPC codes, the node repair, which involves the regeneration of $\lceil n/m \rceil$ symbol(s), requires at most $(d_c-1)\lceil n/m \rceil$ XOR operations if local nodes are used. However, in our repair algorithm, additional operations, which require specific data structures are performed. In particular, three hashmap structures are used. The first one is used to hold the recovery equations for repairing a specific lost symbol, and the second one is used for keeping the sizes of these equations. Moreover, the third hashmap is used to map a symbol to a recovery equations that it is included.} {Initially, the recovery equations, which can be used to repair a lost symbol without BS usage, are identified. 

{There exist $d_v$ recovery equation alternatives for a lost symbol, each of which has at most $d_c$  symbols. In addition, to classify a recovery equation in terms of whether it requires BS communication or not, each of the lost $\mathcal L$ or alive $[n] \setminus \mathcal L$ symbols should be compared with at most $d_c$  terms forming the recovery equations}. This operation requires at most $\mathrm{O}(\frac{nd_vd_c}{m}  \min(|[n]\setminus \mathcal{L}|,|\mathcal{L}|))$ comparisons where $[n]=\{1,2,...,n\}$. {During this identification operation, a min-heap structure is used to define a recovery equation with the minimum cardinality for each lost symbol.} Before using the recovery operations to repair lost symbols, the size of heap is set to $\lceil \frac{n}{m}\rceil$.   Constructing this heap requires $\mathrm{O}(\frac{n}{m}(\log(\frac{n}{m})))$ steps. Each symbol repair requires at most $\mathrm{O}(d_c-1)$ XOR operations. After each symbol repair, at most  $d_vd_c$ updates are performed in the second hashmap  by utilizing other hashmaps to {assign the new  cardinalities of equations in $\mathcal{R}_1^{''}$ which is required for identifying the current needs of helper symbols.} Besides, the min-heap structure is updated at most $\min(\frac{n}{m},d_vd_c)$ times ({due to at most $d_vd_c$ items are affected and the heap size is at most $\lceil n/m \rceil$}) after each symbol repair  as well as  at most $\mathrm{O}(\frac{nd_v}{m})$ search operations  to find the recovery equation having the minimum size  for each lost symbol  are performed before this update. Consequently, all repair operations to repair $\lceil n/m \rceil$ symbols require at most $\mathrm{O}(\frac{nd_vd_c}{m}  \min(|[n]\setminus \mathcal{L}|,|\mathcal{L}|)+d_vd_c\frac{n}{m}+\min(d_vd_c,\frac{n}{m})\frac{n}{m}\log(n/m) +{(n/m)}^2d_v)$ comparisons. Therefore, increasing blocklengths of an LDPC code for a fixed storage node count increases the repair complexity. It should be noted that the inequalities $d_v,d_c \ll n$ generally hold for LDPC code applications.  Since these operations should be repeated $\frac{F}{k}$ times, the node repair operation involves at most 
{$\mathrm{O}(\frac{F}{k}(\frac{nd_vd_c}{m}  \min(|[n]\setminus\mathcal{L}|,|\mathcal{L}|)+ \min(d_vd_c,\frac{n}{m})\frac{n}{m}\log(n/m)+d_vd_c\frac{n}{m}+{(n/m)}^2d_v))$} comparisons {(Here, if we keep the copy of $\mathcal{R}_1$ in every iteration the factor of $\frac{F}{k}$ can be eliminated from the previous complexity expression.)} and  $\mathrm{O}(\frac{F}{k} \frac{n}{m}d_c)$ XOR operations.}

{In traditional $(n,k)$ RS codes, repairing $n/m$ symbols {with a naive method} requires $\mathrm{O}(k^3+k^2+ (n/m)k)$ finite field addition and multiplication operations. Hence, repairing a node involves $\mathrm{O}(k^3+Fk+ F(n/m)) $ operations.} {In MBR codes \cite{Rashmi11}, for repairing a lost node, each helper nodes computes $\mathrm{O}(\frac{F}{B_1}d(n/m))$ finite field multiplications and additions. The newcomer calculates  $\mathrm{O}((n/m)d^3)$ finite field multiplications and additions for inverting the sub-generator matrix, and $\mathrm{O}(d^2)$ multiplications for each $\frac{n}{m}$ symbols. Thus, repairing a lost node requires at most $\mathrm{O}((d^3+\frac{F}{B_1}d^2)\frac{n}{m})$ finite field operations.}
In MSR codes specified in \cite{Rashmi11}, repairing a lost node requires {$\mathrm{O}(\frac{F}{B_2}\alpha(n/m))$} finite  field additions and multiplications from each of the $d$ helper nodes and newcomer employs  $\mathrm{O}(d^3)$ finite field operations for inverting a submatrix as well as $\mathrm{O}(d^2)$ finite field addition and multiplication operations for each of $n/m$ symbols. {For repairing a lost node, similar to MBR codes, $\mathrm{O}((d^3+d^2\frac{F}{B_2})\frac{n}{m})$ finite field addition and multiplication operations are performed in MSR codes \cite{Rashmi11}.  {Hence, for moderate values of $d_v,d_c$, $\min(|[n]\setminus \mathcal{L}|,|\mathcal{L}|)$ and $n/m$,  Greepair algorithm ensures practical use by achieving low time complexity.} }

\subsection{Simulation Results} \label{sec:sims}
\subsubsection{Comparisons with other codes}
In this section, the cost of node repair  is analyzed both theoretically and by simulations in terms of the number of symbols downloaded from the BS as well as from the local network nodes. 
The simulations are conducted on the system model given in Section III. In our simulations, a file size of 128 KB is encoded, and the failures are repaired based on lazy repair with varying $\Delta$ time frames. Our results are normalized by the corresponding stored content size of an RS storage node, i.e., $F_{\mathrm{RS_{node}}}=\frac{F~.~n_{\mathrm{RS}}}{k_{\mathrm{RS}}~.~m_{\mathrm{RS}}}$. We also include our theoretical {average cost expressions} (upper bounds (UB) for LDPC codes and exact expressions for the other codes) in order  to  support our numerical results.

{The cost of node repair for systems using  MBR, MSR, RS, and LDPC codes having code rate $1/2$  is depicted in Fig. \ref{figcost12}}. We use $n=24,~k=12,~m=24$ for RS codes, MBR and MSR codes with $d=23$. {In addition, an LDPC code \cite{Eleftheriou2002} is used with parameters $n=908,~k=454$ and $m=24$}. {As for $(\rho_{\mathrm{D2D}}, \rho_{\mathrm{BS}})$, we have used the pairs $(1,1.2),~(1,12),~(1,17),~(1,26)$. Accordingly, the minimum cost is achieved by MBR codes followed by LDPC codes, MSR codes and RS codes when $\rho_{\mathrm{BS}}\leq 17$ and $\Delta \leq 0.4$. When  $\rho_{\mathrm{BS}} > 17$, RS codes begin outperforming LDPC codes. When  $\rho_{\mathrm{BS}}$ gets values higher than $26$, RS codes outperforms MBR codes in terms of bandwidth cost for some $\Delta$ values of interest.}
The highest count of symbol downloads from other nodes occurs in RS codes, and these results are followed by that of {MBR} codes. Moreover, LDPC, MSR, and MBR codes use $88\%$, {$88\%$} and $86\%$ fewer symbols, respectively, from other network nodes compared to that of RS codes.
\begin{figure}[!htb] 
\includegraphics[width=0.99\textwidth]{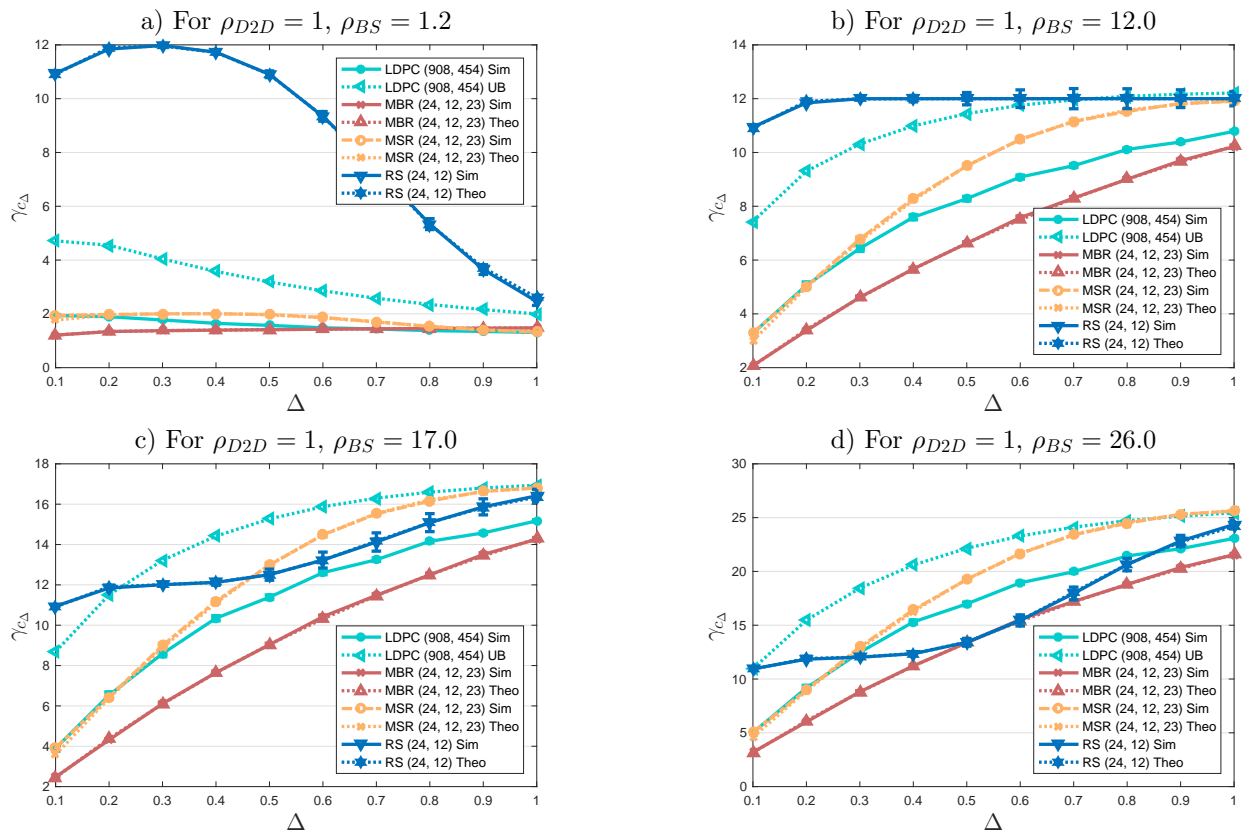} 
 \centering 
 \caption{The repair bandwidth costs for rate 1/2.}\label{figcost12}
\end{figure}

On the other hand,  the highest BS usage, in general, occurs when MSR codes are employed, followed by LDPC codes for $R=1/2$. The minimum BS usage is achieved by RS codes, followed by MBR codes.   It should be noted that MSR codes, RS codes and LDPC codes have the same storage overhead, whereas MBR codes have the highest storage space consumption amongst all of these erasure codes. On average, RS codes require $61\%$, $56\%$, {$66\%$} less BS interaction than LDPC, MBR, and MSR codes do, respectively, for  $R=1/2$.

In Fig. \ref{figcost34}, the node repair costs of systems using MSR, MBR, RS, and LDPC codes, having a code rate of roughly $3/4$  are shown. Parameters $n=24,~k=18, ~m=24$ are used for both RS and MBR codes with $d=23$. {In the simulations, we compare the cost of communication to repair a lost node of the codes having the same rate, i.e., having the same $\frac{k}{n}$ as in \cite{park2017ldpc}. However, notice that high rate MSR codes specified in \cite{EMSRHighRate}, could not exactly achieve the rate 3/4. Likewise, LDPC codes in \cite{EMSRHighRate} could not have the rate 18/23.  Hence, we use parameters  $n=23,~d=22,~k=18, ~m=23$ for the  MSR codes. In addition, {since these MSR codes have fewer storage nodes and consequently, have lower storage overhead in the whole system than the other codes have, to compare these codes on a more fair basis with the other codes}, we scale the results obtained for MSR  codes having rate 18/23 by $\frac{23}{24}=\frac{m_{\mathrm{MSR}}}{m_{\mathrm{RS}}}$.}
 {We use $(n=2056,~k=1542)$  LDPC code with $m=24$. As for $(\rho_{\mathrm{D2D}}, \rho_{\mathrm{BS}})$, the pairs $(1,3),~(1,18),~(1,24), ~(1,50)$ are utilized. When $\rho_{\mathrm{BS}}>3$, MSR codes become more costly than that of LDPC codes for small values of $\Delta$.  When $\rho_{\mathrm{BS}}>24$, RS codes become less costly than that of LDPC and MSR codes in some selected values of $\Delta$. However, RS codes can not outperform MBR codes even for high values of $\rho_{\mathrm{BS}}$ when $\Delta \geq 0.3$, mainly due to the total amount of  BS communication of MBR codes is less than that of RS codes.}
 \begin{figure}[htb!] 
	\centering
\includegraphics[width=\textwidth]{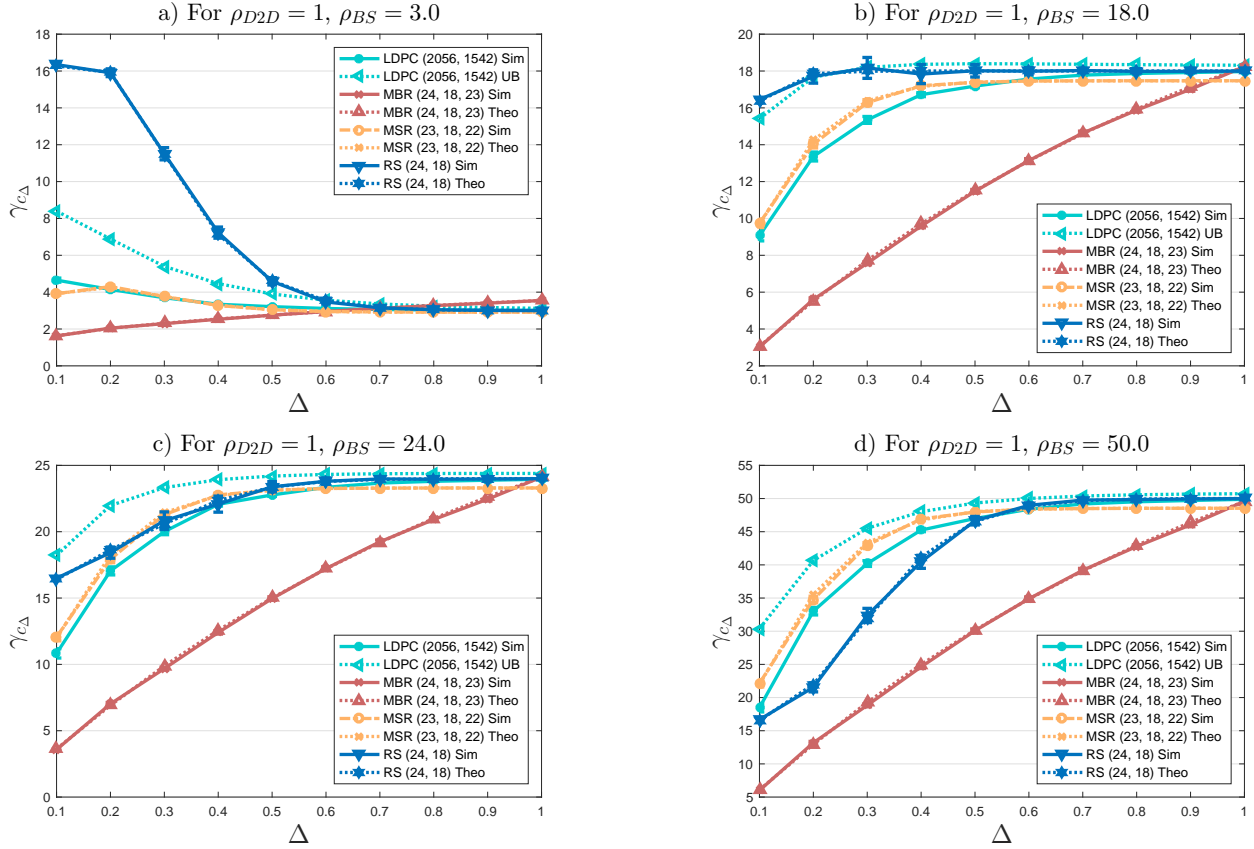} 
 \caption{The repair bandwidth costs for rate 3/4. }\label{figcost34}
\end{figure}
RS codes require the highest D2D bandwidth as long as $\Delta < 0.6$. When $\Delta \leq 0.3$, RS codes are followed by LDPC codes in terms of the number of symbols downloaded from other local network nodes and LDPC codes are followed by MSR and MBR codes. For  $\Delta > 0.6$, MSR codes require the least assistance from other network nodes.

It is clear from these results {that} MBR codes outperform other codes in terms of communication cost  {at the expense of having the highest storage overhead}. The highest BS communication occurs with MSR codes for $\Delta \leq 0.6$, whereas the differences are not dramatic for $\Delta>0.6$.  In the case of MBR codes, we download $34\%$ and $35\%$ fewer symbols from BS as compared to  LDPC and MSR codes, respectively, for the  code rate $3/4$. For $\Delta>0.2$,  MBR codes require $21\%$ fewer symbols from the BS than that of RS codes. 

The lowest BS communication is achieved using RS and MBR codes  for code rates $1/2$  and $3/4$, respectively. Besides, the lowest D2D communication is obtained when  MSR and LDPC codes are utilized for  $R=3/4$ \footnote{{Since lazy repair operations are executed at the end of the interval $\Delta$, it is likely to have no storage node departures for small values of $\Delta$,  which is taken into account by the expression  $b_i(m,p)$. In addition, when $i=m$, the number of symbols downloaded from the local nodes as well as from BS is equal to zero, which reduces the expected number of symbol downloads.}}.

{{Consequently, when a coded caching system is used which defines its crucial parameters as both low storage overhead, low time complexities of  encoding, reconstruction, LDPC codes can be a really good candidate for the low ratios of $\frac{\rho_{\mathrm{BS}}}{\rho_{\mathrm{D2D}}}$ (more specifically,  $\frac{\rho_{\mathrm{BS}}}{\rho_{\mathrm{D2D}}}\leq17$ for $R=1/2$, and $\frac{\rho_{\mathrm{BS}}}{\rho_{\mathrm{D2D}}}\leq24$ for $R=3/4$) among these  codes.} On the other hand, if the value of  $\frac{\rho_{\mathrm{BS}}}{\rho_{\mathrm{D2D}}}$ is high, the option of MBR codes would be preferable. But, these class of codes achieve this performance at the expense of increased storage overhead which may not be feasible from a total cost of ownership point of view for modern cellular systems backed by distributed data storage and caching.}
\subsubsection{LDPC codes with different blocklengths}
{In an attempt to compare the performance of the LDPC codes having different blocklengths, five different options are investigated for the code rate  $3/4$. {We have chosen $(248,186)$, $(424,318)$, $(1096,822)$, $(2168,1626)$ and
$(4024,3018)$ LDPC codes to be used over $m=25$ nodes. As for the pairs} $(\rho_{\mathrm{D2D}},\rho_{\mathrm{BS}})$, we have used $(1,1.2), (1,3),(1,16)$}.
 As shown in
 Fig. \ref{figrevlength34},
for the small values of the ratio, $\frac{\rho_{\mathrm{BS}}}{\rho_{\mathrm{D2D}}}$, the codes having small blocklengths involve slightly higher cost.
However, when this ratio becomes higher, the costs of the codes having different blocklengths become very similar. Besides, it should be noticed that the complexity of repair operation increases with growing blocklength. In  Fig. \ref{figrevlength34}-d), {we have provided}.  $\gamma_{c_\Delta}$ {as a function of} $\rho_{\mathrm{BS}}$ for $\Delta=0.5$. 
\begin{figure}[!h] 
	\centering
\includegraphics[width=0.95\textwidth]{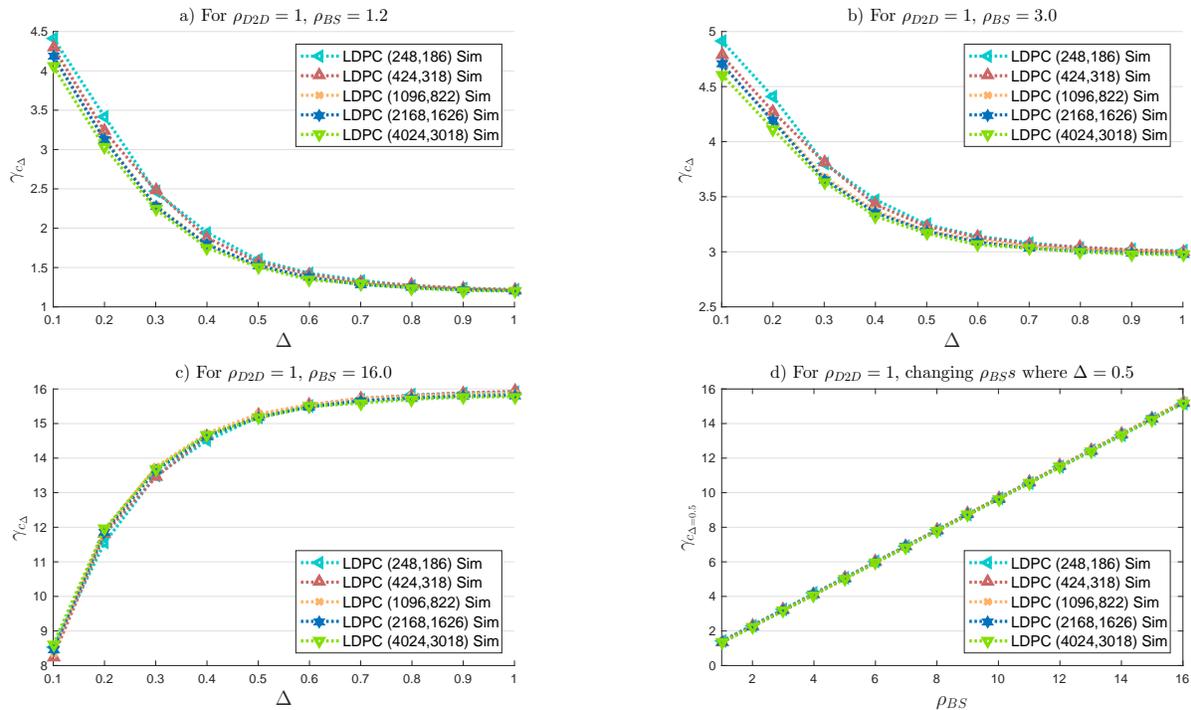} 
 \caption{The repair cost for different LDPC blocklengths for rate 3/4. }\label{figrevlength34}
\end{figure}

{Notice that when $\Delta$ increases, the number of BS enabled repairs increases because, in this case, the availability of helper symbols in local nodes is decreased. Let  $\bar{|r|}$ be the  expected number of downloaded helper symbols used to repair a lost symbol through local nodes.} { When $\rho_{\mathrm{BS}}/\rho_{\mathrm{D2D}}<\bar{|r|}$ is satisfied, the BS utilization in the repair process is advantageous in terms of node repair bandwidth cost. More explicitly, if a lost symbol is repaired through BS, only the original lost symbol is downloaded. However, if this lost symbol is repaired using local nodes only, then the symbols in the recovery equation should be downloaded, which would require $\bar{|r|}$ downloads in average from the other nodes.  Consequently, when $\rho_{\mathrm{BS}}/\rho_{\mathrm{D2D}}$ is less than $\bar{|r|}$, the total node repair bandwidth cost decreases with growing $\Delta$ since the BS utilization in node repairs increases with growing $\Delta$ as can be seen in Fig.  \ref{figrevlength34}-a) and Fig.  \ref{figrevlength34}-b). In contrast,  when $\rho_{\mathrm{BS}}/\rho_{\mathrm{D2D}}$ is higher than $\bar{|r|}$, the total node repair bandwidth cost increases with growing $\Delta$ as can be seen in Fig.  \ref{figrevlength34}-c). }

\subsubsection{Comparisons with the optimum solution}
In order to conceive the performance gap between optimum solutions and the proposed Greepair Algorithm, we have compared the proposed algorithm with  two different optimum strategies. The first optimum strategy, namely Opt.-1, aims to guarantee the lowest possible base station communication and ensure this condition by providing  the minimum possible device communication. The second optimum strategy, namely Opt.-2, aims to guarantee the overall lowest communication cost. For this strategy,  we have calculated $\gamma_{c_{\Delta}}$  with two different pairs for $(\rho_{\mathrm{D2D}},\rho_{\mathrm{BS}})$ which are $(\rho_{\mathrm{D2D}}=1,\rho_{\mathrm{BS}}=10)$ and $(\rho_{\mathrm{D2D}}=1,\rho_{\mathrm{BS}}=20)$. 

In our results, we have used $n=184, ~k=138$ for LDPC codes.
\begin{table}[!t]
\centering
 \begin{tabular}{|c|c|c|c|c|c|c|c|c|c|c|c|c|}
 \hline
\multicolumn{1}{|c|}{\multirow{2}{*}{\textbf{$\boldsymbol{n,k}$}}} & {\multirow{2}{*}{{$\boldsymbol{\lceil n/m \rceil}$}}}&{\multirow{2}{*}{{$\boldsymbol{\Delta}$}}} &\multicolumn{2}{c|}{\textbf{Opt.- 1}} & \multicolumn{2}{c|}{\textbf{Opt.-2}}  \\
\cline{4-7}
&&\multicolumn{1}{|c|}{}&  $\tau_{\Delta}$& $\phi_{\Delta}$ &$\boldsymbol{\rho_{D2D}=1,\rho_{BS}=10}$&$\boldsymbol{\rho_{D2D}=1,\rho_{BS}=20}$\\
  \hline \hline
 \multicolumn{1}{|c|}{\parbox[c]{15mm}{\multirow{6}{=}{~~~~~\rotatebox[origin=c]{90}{~~~$\boldsymbol{n=184, ~k=138~~~~~}$}}}} &  \parbox[t]{15mm}{\multirow{3}{*}{~~~~~\rotatebox[origin=c]{90}{$~~\boldsymbol{\lceil n/m \rceil=3}~$}}}& \textbf{$\boldsymbol{\Delta=0.1}$  }  & 2.601 \%& 0.569\%  & 1.927 \%  & \multicolumn{1}{c|}{1.589\%}   \\
 \cline{3-7}
  & & \textbf{$\boldsymbol{\Delta=0.4}$ } & 0.012 \% &  0.130\%  & 0.117 \% & \multicolumn{1}{c|}{0.123\%}\\
      
\cline{3-7}
   &   &  \textbf{$\boldsymbol{\Delta=0.7}$ } & 0\% & 0\%  &0\% & \multicolumn{1}{c|}{0\%}\\
   \cline{2-7}     \cline{2-7}
      \cline{2-7}     \cline{2-7}
&\parbox[t]{15mm}{\multirow{3}{*}{~~~~~\rotatebox[origin=c]{90}{~~~$~\boldsymbol{\lceil n/m\rceil=6}~$}}}& \textbf{$\boldsymbol{\Delta=0.1}$  } & 8.298\% & 1.871\% & 6.044\%& \multicolumn{1}{c|}{4.998\%}\\
 \cline{3-7}
    &  &  \textbf{$\boldsymbol{\Delta=0.4}$ } & 2.222\% & 0.194\%  & 0.375\%   &  \multicolumn{1}{c|}{0.289\%}\\
      
\cline{3-7}
     & &  \textbf{$\boldsymbol{\Delta=0.7}$ }& 1.201\% & 0\% & 0\% &  \multicolumn{1}{c|}{0\%}\\
      \hline\hline
 \end{tabular}
 \caption{ \label{tab:OptimalLPPCImprovements} The percentages of the improvements obtained by optimum strategies over Greepair. }
\end{table}
Since  finding the optimum solution for LDPC node repair problem is shown to be a special case of an NP-Complete problem in \cite{HaytaogluArxiv}, we used exhaustive search. Moreover, up to $(\lceil n/m \rceil)!(d_v+1)^{\lceil n/m \rceil}$ different options exist for repairing a lost node to find the optimum cost value,  moderate $\lceil n/m \rceil$ values such as $\lceil n/m \rceil=3$,   and $\lceil n/m \rceil=6$ are used in numerical results.  

In Table \ref{tab:OptimalLPPCImprovements}, the   improvements in terms of percentages  achieved by optimum strategies over the Greepair algorithm are demonstrated. More explicitly, in Table \ref{tab:OptimalLPPCImprovements}, the improvement of Opt.-1 strategy in terms of decrement on $\tau_{\Delta}$ and $\phi_{\Delta}$ according to Greepair and the improvement of  Opt.-2 strategy over Greepair  in terms of $\gamma_{c_{\Delta}}$ are shown. The improvements are observed to be increasing with growing $\lceil n/m \rceil$ and especially for $\Delta=0.1$. Since there are many node failures  for  high values of $\Delta$, the advantage of the optimum strategy is less prominent as expected.
\section{Conclusion}
In this paper, we investigated the possibility of LDPC code utilization within the context of coded caching {with appropriate repair and regeneration}. To that end, the other well-known erasure coding techniques such as RS, MBR, and MSR codes are  analyzed with respect to their bandwidth usage from the BS and other network nodes. In order to achieve a node repair with minimal cost, a novel greedy algorithm called Greepair is proposed. According to the theoretical and simulation results, if the gap between  ${\rho_{\mathrm{BS}}}$ and ${\rho_{\mathrm{D2D}}}$ is not dramatically high, {LDPC codes can be a very reasonable choice due to  their low encoding and decoding (reconstruction) complexities as well as their low storage overhead.} As a future work, special irregular designs for LDPC codes, which would further reduce the BS requirements in the repair process, will be considered.


%
\section*{Acknowledgment}
This study is supported by TUBITAK under grant number 119E235 and Pamukkale University Scientific Research Projects Department under grant number 2018FEBE009. {We thank the anonymous reviewers for their suggestions and comments.}

\ifCLASSOPTIONcaptionsoff
  \newpage
\fi

\bibliographystyle{IEEEtranTCOM}
\bibliography{paper}
\end{document}